\documentclass[twocolumn]{article}
\usepackage{amssymb,amsmath,latexsym}

\newtheorem{theorem}{Theorem}[section]
\newtheorem{lemma}[theorem]{Lemma}

\newtheorem{definition}[theorem]{Definition}
\newtheorem{example}{Example}

\newenvironment{proof}{\noindent
  \textbf{Proof.}}{\hfill$\Box$\\}
\newenvironment{proofidea}{\noindent
  \textbf{Proof idea.}}{\hfill$\Box$\\}
\newenvironment{proofsketch}{\noindent
  \textbf{Proof sketch.}}{\hfill$\Box$\\}

\newcommand{\CMATELCDLT}{\textbf{CMATEL(CD+LT)}}
\newcommand{\cut}[1]{}

\newcommand{\instr}[5]{\ensuremath{\hbox to 60 pt
    {${#1}$\hfil${#2}$\hfil$ \rightarrow
      $\hfil${#3}$\hfil${#4}$\hfil${#5}$}}}

\newcommand{\D}{\ensuremath{\Delta}}
\newcommand{\G}{\ensuremath{\Gamma}}

\newcommand{\vp}{\ensuremath{\varphi}}

\newcommand{\fm}[1]{\emph{#1}}
\newcommand{\de}[1]{\emph{#1}}

\newcommand{\rel}[1]{\ensuremath{\mathcal{#1}}}

\newcommand{\power}[1]{\ensuremath{\mathcal{P}(#1)}}
\newcommand{\powerne}[1]{\ensuremath{\mathcal{P}^{\tiny +}(#1)}}
\newcommand{\union}{\, \cup \,}

\newcommand{\bigunion}{\bigcup \,}
\newcommand{\biginter}{\bigcap \,}

\newcommand{\inter}{\, \cap \,}

\newcommand{\crh}[2]{\ensuremath{\{\, #1 \mid \, #2\, \}}}

\newcommand{\set}[1]{\ensuremath{ \{ #1 \} }}

\newcommand{\nat}{\ensuremath{\mathbb{N}}}

\newcommand{\CMAELCD}{\textbf{CMAEL(CD)}}

\newcommand{\lang}{\ensuremath{\mathcal{L}}}

\newcommand{\CTL}{\textbf{CTL}}

\newcommand{\LTL}{\ensuremath{\textbf{LTL}}}

\newcommand{\con}{\wedge}

\newcommand{\equivalence}{\leftrightarrow}

\newcommand{\ap}{\textbf{\texttt{AP}}}

\newcommand{\truth}{\ensuremath{\top}}

\newcommand{\next}{\!\raisebox{-.2ex}{ 
            \mbox{\unitlength=0.9ex
            \begin{picture}(2,2)
            \linethickness{0.06ex}
            \put(1,1){\circle{2}} 
            \end{picture}}}       
            \,}
          \newcommand{\until}{\ensuremath{\hspace{2pt}\mathcal{U}}}

\newcommand{\kframe}[1]{\ensuremath{\mathfrak{#1}}}

\newcommand{\mmodel}[1]{\ensuremath{\mathcal{#1}}}
\newcommand{\hintikka}[1]{\ensuremath{\mathcal{#1}}}

\newcommand{\sat}[3]{\ensuremath{\mmodel{#1}, #2 \Vdash #3}}

\newcommand{\notsat}[3]{\ensuremath{\mmodel{#1}, #2 \nVdash #3}}

\newcommand{\Rule}[1]{\textbf{(#1)}}

\newcommand{\subf}[1]{\ensuremath{\mathsf{Sub}(#1)}}

\newcommand{\agents}{\ensuremath{\Sigma}}

\newcommand{\st}[1]{\ensuremath{\mathbf{states}(#1)}}

\newcommand{\brancharrow}{\ensuremath{\Longrightarrow}}

\newcommand{\tableau}[1]{\ensuremath{\mathcal{#1}}}

\newcommand{\knows}[1]{\ensuremath{\mathbf{K}}_{#1}}

\newcommand{\commonk}[1]{\ensuremath{\mathbf{C}}_{#1}}
\newcommand{\distrib}[1]{\ensuremath{\mathbf{D}}_{#1}}

\newcommand{\psmodel}[1]{\ensuremath{\mathcal{M}^*}}
\newcommand{\pseudomodel}[1]{\ensuremath{\mathcal{M}^{**}}}

\begin{document}

\title{Tableau-based decision procedure for full coalitional
  multiagent temporal-epistemic logic of linear time}

\author{Valentin Goranko\footnote{School of Mathematics,
       University of the Witwatersrand, South Africa} \and Dmitry
     Shkatov\footnote{School of Computer Science,
       University of the Witwatersrand, South Africa}}

\date{}
   
\maketitle

\begin{abstract}
  We develop a tableau-based decision procedure for the full
  coalitional multiagent temporal-epistemic logic of linear time
  \CMATELCDLT. It extends \LTL\ with operators of common and
  distributed knowledge for all coalitions of agents. The tableau
  procedure runs in exponential time, matching the lower bound
  obtained by Halpern and Vardi for a fragment of our logic, thus
  providing a complexity-optimal decision procedure for \CMATELCDLT.
\end{abstract}

\section{Introduction}
\label{sec:into}

Knowledge and time are among the most important aspects of multiagent
systems. Various \emph{temporal-epistemic logics}, proposed as logical
frameworks for reasoning about these ascpects of multiagent systems
were studied in a number of publications during the 1980's, eventually
summarized in a uniform and comprehensive study by Halpern and Vardi
\cite{HV89}. In~\cite{HV89}, the authors considered several essential
characteristics of temporal-epistemic logics: \emph{one vs. several
  agents}, \emph{synchrony vs. asynchrony}, \emph{(no) learning},
\emph{(no) forgetting}, \emph{linear vs.  branching time}, and the
(non-) existence of a \emph{unique initial state}. Based on these,
they identify and analyze 96 temporal-epistemic logics and obtain
lower bounds for the complexity of a satisfiability problem in each of
them. It turns out that most of the logics with more than one agent
who do not learn or do not forget, are undecidable (with common
knowledge), or decidable but with non-elementary time lower bound
(without common knowledge). For the remaining multiagent logics, the
lower bounds from~\cite{HV89} range from PSPACE (systems without
common knowledge), through EXPTIME (with common knowledge), to
EXPSPACE (syn\-chro\-nous systems with no learning and unique initial
state). To the best of our knowledge, however, even for the logics
from~\cite{HV89} with a relatively low 
complexity lower bound, no decision procedures with matching
upper bounds have been developed. In this paper, we set out to develop
such decision procedures based on incremental tableaux, starting with
the multiagent case over linear time, which involves no essential
interaction between knowledge and time. It turns out that, under no
other assumptions regarding learning or forgetting, the synchronous
and asynchronous semantics are equivalent with respect to
satisfiability. We consider a more expressive epistemic language than
the ones considered in \cite{HV89}, to wit, the one involving
operators for common and for distributed knowledge for \emph{all
  coalitions} of agents. We call the resulting logic \CMATELCDLT\
(``\textit{\textbf{C}oalitional \textbf{M}ulti-\textbf{A}gent
  \textbf{T}emporal \textbf{E}pistemic \textbf{L}ogic with operators
  for \textbf{C}ommon and \textbf{D}istributed knowledge and
  \textbf{L}inear \textbf{T}ime}''). The decision procedure for
satisfiability in \CMATELCDLT\ developed herein runs in exponential
time, which together with the lower bound for the fragment of
\CMATELCDLT\ obtained in~\cite{HV89}, implies EXPTIME-completeness of
\CMATELCDLT.

\section{The Logic \\ CMATEL(CD+LT)}
\label{sec:ETLTL_CD}

\subsection{Syntax}
\label{sec:syntax}

The language \lang\ of \CMATELCDLT\ contains a set \ap\ of atomic
propositions, a sufficient repertoire of Boolean connectives, say
$\neg$ (``not'') and $\con$ (``and''), the temporal operators $\next$
(``next'') and $\until$ (``until'') of the logic \LTL, as well as the
epistemic operators $\distrib{A} \vp$ (``\emph{it is distributed
  knowledge among agents in $A$ that $\vp$}''), and $\commonk{A} \vp$
(``\emph{it is common knowledge among agents of $A$ that $\vp$}'') for
every non-empty $A \subseteq \agents$, where $\agents$ is the set of
names of agents belonging to \lang. The set $\agents$ is assumed to be
finite and non-empty; its subsets are called \emph{coalitions} (of
agents).  Thus, the formulae of \CMATELCDLT\ are defined as follows:
$$\vp := p \mid \neg \vp \mid (\vp_1 \con \vp_2) \mid \next \vp
\mid (\vp_1 \until \vp_2) \mid \distrib{A} \vp \mid \commonk{A} \vp$$
where $p$ ranges over \ap\ and $A$ ranges over the set of non-empty
subsets of $\agents$, henceforth denoted \powerne{\agents}.  We write
$\vp \in \lang$ to mean that $\vp$ is a formula of $\lang$.

The operators of individual knowledge $\knows{a} \vp$, where $a \in
\agents$ (``\emph{agent $a$ knows that \vp}''), can then be defined as
$\distrib{\set{a}} \vp$, henceforth written $\distrib{a} \vp$.  The
other Boolean and temporal connectives can be defined as usual.  We
omit parentheses when this does not result in ambiguity.

Formulae of the form $\neg \commonk{A} \vp$ are \emph{epistemic
  eventualities}, while those of the form $\vp \until \psi$ are
\emph{temporal eventualities}. 

$\nat = \{0,1,\ldots \}$ denotes the set of natural numbers.

\subsection{Semantics}
\label{sec:semantics}

\begin{definition}
  \label{def:stes}
  A \emph{temporal-epistemic system} (TES) is a
  tuple $\kframe{G} = (\agents, S, R, \set{\rel{R}^D_A}_{A \in
    \powerne{\agents}}, \set{\rel{R}^C_A}_{A \in
    \powerne{\agents}})$, where: 
  \begin{enumerate}
    \itemsep-2pt
  \item $\agents$ is a finite, non-empty set of \de{agents};
  \item $S \ne \emptyset$ is a set of \de{states};
  \item $R$ is a non-empty set of \de{runs}; where each $r \in R$ is a
    function $r: \nat \mapsto S$. A pair $(r, n)$, where $r \in R$ and
    $n \in \nat$, is called a \de{point}. The set of all points in
    $\kframe{G}$ is denoted $P(\kframe{G})$. Every point $(r, n)$
    represents the state $r(n)$; note, however, that different points
    can represent the same state.
  \item for every $A \in \powerne{\agents}$, $\rel{R}^D_A$ and
    $\rel{R}^C_A$ are binary relation on P(\kframe{G}), such that
    $\rel{R}^C_A$ is the reflexive and transitive closure of
    $\bigunion_{A' \subseteq A} R^D_{A'}$.
  \end{enumerate}
\end{definition}

\begin{definition}
  \label{def:stefs}
  A \de{temporal-epistemic frame} (TEF) is a TES $\kframe{G} =
  (\agents, S, R, \set{\rel{R}^D_A}_{A \in \powerne{\agents}},
  \set{\rel{R}^C_A}_{A \in \powerne{\agents}})$, where each
  $\rel{R}^D_A$ is an equivalence relation satisfying the following
  condition: (\dag) $\rel{R}^D_A = \biginter_{a \in A}
  \rel{R}^D_{\set{a}}$.  If condition (\dag) is replaced by the
  following: (\dag \dag) $\rel{R}^D_A \subseteq \rel{R}^D_B \text{
    whenever } B \subseteq A$, then \kframe{F} is a
  \de{temporal-epistemic pseudo-frame} (pseudo-TEF).
\end{definition}

Notice that, in (pseudo-)TEFs, $\rel{R}^C_A$ is the transitive closure
of $\bigunion_{a \in A} R^D_{\set{a}}$ and, thus, an equivalence
relation.

\begin{definition}
  \label{def:cmaem}
  A \de{temporal-epistemic model} (TEM, for short) is a tuple
  $\mmodel{M} = (\kframe{F}, L)$, where
  \begin{description}
    \itemsep-4pt
  \item[($i$)] $\kframe{F}$ is a TEF with a set of runs $R$;
  \item[($ii$)] $L: R \times \nat\ \mapsto \power{\ap}$
    is a \de{labeling function}, where $L(r,n)$ is the set of atomic
    propositions true at $(r,n)$.
  \end{description}
  If the condition (i) is replaced by the requirement that
  $\kframe{F}$ is a pseudo-TEF, then \mmodel{M} is a 
  \de{temporal-epistemic pseudo-model} (pseudo-TEM).
\end{definition}

A TES $\kframe{G}$ is called \emph{synchronous} if for every $A \in
\powerne{\agents}$, if $((r, n), (r',n')) \in \rel{R}^D_A$, then $n =
n'$. Synchronous temporal-epistemic (pseudo)-models are defined
accordingly. Hereafter we consider the general case, but all
definitions and results apply likewise to the synchronous case, unless
stated otherwise. The tableau construction can accommodate the
synchronous case at no extra cost and eventually we show that, under
no other assumptions, the presence or absence of synchrony does not
affect the satisfiability of formulae.

\begin{definition}
  \label{def:satisfaction}
  The satisfaction of formulae at points in 
  (pseudo-)TEMs is defined as follows:

\sat{M}{(r, n)}{p} iff $p \in L(r, n)$;

\sat{M}{(r, n)}{\neg \vp} iff not \sat{M}{(r, n)}{\vp};

\sat{M}{(r, n)}{\vp \con \psi} iff \sat{M}{(r, n)}{\vp} and
\sat{M}{(r, n)}{\psi};

\sat{M}{(r, n)}{\next \vp} iff \sat{M}{(r, n+1)}{\vp};

\sat{M}{(r, n)}{\vp \until \psi} iff $\sat{M}{(r, i)}{\psi}$ for some
$i \geq n$ \newline \indent ~ \hspace{1cm} such that \sat{M}{(r,
  j)}{\vp} for every $n \leq j < i$;

\sat{M}{(r, n)}{\distrib{A} \vp} iff \sat{M}{(r',n')}{\vp} \newline
\indent ~ \hspace{1cm} for every $((r, n), (r', n')) \in \rel{R}^D_A$;

\sat{M}{(r, n)}{\commonk{A} \vp} iff \sat{M}{(r',n')}{\vp} \newline
\indent ~ \hspace{1cm} for every $((r, n), (r', n')) \in \rel{R}^C_A$;
\end{definition}

Note, that in the semantics defined above \emph{the labelling function
  acts on points, not states}, i.e., it is \emph{point-based}. To make
the semantics \emph{state-based}, one must impose the additional
condition: if $r(n) = r'(n')$ then $L(r, n) = L(r', n')$. However, for
the case of linear time logics these two semantics are equivalent in
terms of satisfiability and validity (this is an easy consequence of
the fact that, in the linear case, all epistemic operators have
built-in implicit universal quantification over paths).

The satisfaction condition for the operator $\commonk{A}$ can be
paraphrased in terms of reachability.  Let $\kframe{F}$ be a (pseudo-)
TEF over the set of runs $R$ and let $(r, n) \in R \times \nat$.  We
say that a point $(r', n')$ is \de{$A$-reachable from $(r, n)$} if
either $r = r'$ and $n = n'$ or there exists a sequence $(r, n) =
(r_0, n_0), (r_1, n_1), \ldots, (r_{m-1}, n_{m-1}), (r_m, n_{m}) =
(r', n')$ of points in $R \times \nat$ such that, for every $0 \leq i
< m$, there exists $a_i \in A$ such that $((r_i, n_{i}), (r_{i+1},
n_{i+1})) \in R^D_{a_i}$. Then, the satisfaction condition for
$\commonk{A}$ becomes equivalent to the following:

\sat{M}{(r, n)}{\commonk{A} \vp} iff \sat{M}{(r', n')}{\vp}
  whenever $(r', n')$ is $A$-reachable from $(r, n)$.

Satisfiability and validity in (a class of) models is defined as
usual.

It is easy to see that if $\agents = \set{a}$, then $\distrib{a} \vp
\equivalence \commonk{a} \vp$ is valid in every TEM for every $\vp \in
\lang$. Thus, the single-agent case is essentially trivialized and,
therefore, we assume hereafter that $\agents$ contains at least 2
(names of) agents.
\section{Hintikka structures}
\label{sec:hintikka_structures}

Even though we are ultimately interested in testing formulae of
$\lang$ for satisfiability in a TEM, the tableau procedure we present
tests for satisfiability in a more general kind of semantic structures,
namely a \emph{Hintikka structure}. We will show that $\theta \in
\lang$ is satisfiable in a TEM iff it is satisfiable in a Hintikka structure, hence the latter test is equivalent to the former.  The
advantage of working with Hintikka structures lies in the fact that
they contain as much semantic information about $\theta$ as is
necessary, and no more.  More precisely, while models provide the
truth value of every formula of $\lang$ at every state, Hintikka
structures only determine the truth of formulae directly involved in
the evaluation of a fixed formula $\theta$, in whose satisfiability we
are interested.  Another important difference between models and
Hintikka structures is that, in Hintikka structures the epistemic
relations $\rel{R}^D_A$ and $\rel{R}^C_A$ only have to satisfy the
properties laid down in Definition~\ref{def:stes}.  All the other
information about the desirable properties of epistemic relations is
contained in the labeling of states in Hintikka structures. This
labeling ensures that every Hintikka structure generates a
pseudo-model (by the construction of Lemma~\ref{lem:STEHS_to_psSTEM}),
which can then be turned into a model using the construction of
Lemma~\ref{lem:psSTEM_to_STEM}.

\begin{definition}
  \label{def:fully_expanded}
  A set $\D \subseteq \lang$ is \de{fully expanded} if it satisfies
  the following conditions (\subf{\psi} stands for the set of
  subformulae of $\psi$):
  \begin{enumerate}
    \itemsep-2pt
 \item if $\neg \neg \vp \in \D$ then $\vp \in \D$;
  \item if $\vp \con \psi \in \D$, then $\vp \in \D$ and $\psi \in
    \D$;
  \item if $\neg (\vp \con \psi) \in \D$ then $\neg \vp \in \D$ or
    $\neg \vp \in \D$;
  \item if $\neg \next \vp \in \D$ then $\next \neg \vp \in \D$;
  \item if $\vp \until \psi \in \D$ then $\psi \in \D$ or $\vp, \next
    (\vp \until \psi) \in \D$;  
   \item if $\neg (\vp \until \psi) \in \D$ then $\neg \psi,
    \neg \vp \in \D$ or $\neg \psi, \neg \next (\vp \until \psi) \in
    \D$;
  \item if $\distrib{A} \vp \in \D$ then $\distrib{A'} \vp \in \D$
    for every $A'$ such that $A \subseteq A' \subseteq \agents$;
  \item if $\distrib{A} \vp \in \D$ then $\vp \in \D$;
  \item if $\commonk{A} \vp \in \D$ then $\distrib{a} (\vp \con
    \commonk{A} \vp) \in \D$ for every $a \in A$;
  \item if $\neg \commonk{A} \vp \in \D$ then $\neg \distrib{a} (\vp
    \con \commonk{A} \vp) \in \D$ for some $a \in A$;
  \item if $\psi \in \D$ and $\distrib{A} \vp \in \subf{\psi}$ then
    either $\distrib{A} \vp \in \D$ or $\neg \distrib{A} \vp \in \D$.
  \end{enumerate}
\end{definition}

\begin{definition}
  \label{def:hs}
  A \de{temporal-epistemic Hintikka structure} (TEHS) is a tuple
  $(\agents, S, R, \set{\rel{R}^D_A}_{A \in \powerne{\agents}},
  \set{\rel{R}^C_A}_{A \in \powerne{\agents}}, H)$ such that
  $(\agents, S, R, \set{\rel{R}^D_A}_{A \in \powerne{\agents}},
  \set{\rel{R}^C_A}_{A \in \powerne{\agents}})$ is a TES, and $H$ is a
  labeling of points in $R \times \nat$ with sets of formulae,
  satisfying the following conditions, for all $(r, n) \in R \times
  \nat$:
  \begin{description}
    \itemsep-2pt
  \item[H1] if $\neg \vp \in H(r, n)$, then $\vp \notin H(r, n)$;
  \item[H2] $H(r, n)$ is fully expanded;
  \item[H3] if $\next \vp \in H(r, n)$, then $\vp \in H(r, n+1)$;
  \item[H4] if $\vp \until \psi \in H(r, n)$, then there exists $i
    \geq n$ such that $\psi \in H(r,i)$ and $\vp \in H(r, j)$
    holds for every $n \leq j < i$;
  \item[H5] if $\neg \distrib{A} \vp \in H(r, n)$, then there exists
    $r' \in R$ and $n'\in \nat$ such that $((r,n), (r', n')) \in \rel{R}^D_A$ and $\neg \vp \in H(r', n')$;
  \item[H6] if $((r,n), (r', n')) \in \rel{R}^D_A$, then $\distrib{A'}
    \vp \in H(r, n)$ iff $\distrib{A'} \vp \in H(r', n')$, for every
    $A' \subseteq A$;
  \item[H7] if $\neg \commonk{A} \vp \in H(r, n)$, then there exists
    $r' \in R$ and $n'\in \nat$ such that $((r, n), (r', n')) \in \rel{R}^C_A$ and $\neg \vp \in H(r', n')$. 
    \end{description}
    Synchronous TEHSs (STEHSs) are defined likewise.
\end{definition}

\begin{definition}
\label{def:sat_in_hintikka}
A set of formulae $\Theta$ is \emph{satisfiable} in a TEHS
$\hintikka{H}$ with a labeling function $H$ if there exists a point
$(r, n) \in \hintikka{H}$ such that $\Theta \subseteq H(r,
n)$. Analogously for formulae.
\end{definition}

Now, we show that $\theta \in \lang$ is satisfiable in a TEM iff it
is satisfiable in a TEHS. One direction is almost immediate, as
every TEM naturally induces a TEHS.  More precisely, given a TEM
$\mmodel{M}$, define the \emph{extended labeling} $L^+$ on the set of
points of $\mmodel{M}$ as follows: $L^+(r, n) = \crh{\vp}{\sat{M}{(r,
    n)}{\vp}}$ for every $(r, n)$.  The following claim is then
straightforward.

\begin{lemma}
  \label{lem:form_models_to_hintikka}
  Let $\mmodel{M} = (\kframe{F}, L)$ be a TEM satisfying $\theta \in
  \lang$, and let $L^+$ be the extended labeling on \mmodel{M}.  Then,
  $\hintikka{H} = (\kframe{F}, \ap, L^+)$ is a TEHS satisfying
  $\theta$.
\end{lemma}

For the opposite direction, we first prove that the
existence of a TEHS satisfying $\theta$ implies the existence of a
pseudo-model satisfying $\theta$; then, we show that this in turn
implies the existence of a model satisfying $\theta$.

\begin{lemma}
  \label{lem:STEHS_to_psSTEM}
  Let $\theta \in \lang$ be such that there exists a TEHS for
  $\theta$.  Then, $\theta$ is satisfiable in a pseudo-TEM.
\end{lemma}

\begin{proof}
  Let $\hintikka{H} = (\agents, S,
  R, \set{\rel{R}^D_A}_{A \in \powerne{\agents}}, \set{\rel{R}^C_A}_{A
    \in \powerne{\agents}}, \linebreak H)$ be a TEHS for $\theta$.  We
  build a pseudo-TEM satisfying $\theta$ as follows.  First, for every
  $A \in \powerne{\agents}$, let $\rel{R}'^D_A$ be the reflexive,
  symmetric, and transitive closure of $\bigunion_{A \subseteq B}
  \rel{R}^D_{B}$ and let $\rel{R}'^C_A$ be the transitive closure of
  $\bigunion_{a \in A} \rel{R}'^D_a$. Notice that $\rel{R}^D_A
  \subseteq \rel{R}'^D_A$ and $\rel{R}^C_A \subseteq \rel{R}'^C_A$ for
  every $A \in \powerne{\agents}$. Next, let $L(r, n) = H(r, n) \inter
  \ap$, for every point $(r, n) \in R \times \nat$.  It is then easy
  to check that $\mmodel{M}' = (\agents, S, R, \set{\rel{R}'^D_A}_{A
    \in \powerne{\agents}}, \linebreak \set{\rel{R}'^C_A}_{A \in
    \powerne{\agents}}, \ap, L)$ is a pseudo-TEM.  It is also easy to
  check that the construction preserves synchrony.

  To complete the proof of the lemma, we show, by induction on the
  formula $\chi \in \lang$ that, for every point $(r, n)$ and every
  $\chi \in \lang$, the following hold:

  \textbf{(i)} $\chi \in H(r, n) \text{ implies } \sat{M'}{(r,
    n)}{\chi}$;

  \textbf{(ii)} $\neg \chi \in H(r, n) \text{ implies } \sat{M'}{(r,
    n)}{\neg \chi}$.

  Let $\chi$ be some $p \in \ap$.  Then, $p \in H(r, n)$ implies $p
  \in L(r, n)$ and thus, \sat{M'}{(r,n)}{p}; if, on the other hand,
  $\neg p \in H(r, n)$, then due to (H1), $p \notin H(r, n)$ and thus
  $p \notin L(r, n)$; hence, \sat{M'}{(r, n)}{\neg p}.

  Assume that the claim holds for all subformulae of $\chi$; then, we have
  to prove that it holds for $\chi$, as well.

  Suppose that $\chi = \neg \vp$.  If $\neg \vp \in H(r, n)$, then the
  inductive hypothesis immediately gives us $\sat{M'}{(r, n)}{\neg
    \vp}$; if, on the other hand, $\neg \neg \vp \in H(r, n)$, then by
  virtue of (H2), $\vp \in H(r, n)$ and hence, by inductive
  hypothesis, $\sat{M'}{(r, n)}{\vp}$ and thus $\sat{M'}{(r, n)}{\neg
    \neg \vp}$.

  The cases of $\chi = \vp \con \psi$ and $\chi = \next \vp$ are
  straightforward, using (H2) and (H3).

  Suppose that $\chi = \distrib{A} \vp$.  Assume, first, that
  $\distrib{A} \vp \in H(r, n)$.  In view of the inductive hypothesis,
  it suffices to show that $((r, n), (r', n')) \in \rel{R}'^D_A$
  implies $\vp \in H(r, n)$.  Assuming $((r, n), (r', n')) \in
  \rel{R}'^D_A$, there are two cases to consider.  If $(r,n) = (r',
  n')$, then the conclusion immediately follows from (H2).  Otherwise,
  there exists an undirected path from $(r, n)$ to $(r', n')$ along
  the relations $\rel{R}^D_{A'}$, where each $A'$ is a superset of
  $A$.  Then, due to (H6), $\distrib{A} \vp \in H(r', n')$; hence, by
  (H2), $\vp \in H(r', n')$, as desired.

  Now, let $\neg \distrib{A} \vp \in H(r, n)$. By (H5), there exist
  $r' \in R$ and $n' \in \nat$ such that $((r, n), (r', n')) \in
  \rel{R}^D_A$ and $\neg \vp \in H(r', n')$.  As $\rel{R}^D_A
  \subseteq \rel{R}'^D_A$, the claim follows from the inductive
  hypothesis.

  Suppose that $\chi = \commonk{A} \vp$.  Assume that $\commonk{A} \vp
  \in H(r, n)$.  By inductive hypothesis, it suffices to show that if
  $(r', n')$ is $A$-reachable from $(r, n)$, then $\vp \in H(r',
  n')$. If $(r,n) = (r', n')$ the claim follows from (H2). So, suppose
  for some $m \geq 1$, there exists a sequence of points $(r, n) =
  (r_0, n_0), \ldots, \linebreak (r_{m-1}, n_{m-1}), (r_m, n_{m}) =
  (r', n')$ such that, for every $0 \leq i < m$, there exists $a_i \in
  A$ such that $((r_i, n_{i}),(r_{i+1}, n_{i+1})) \in
  \rel{R}'^D_{a_i}$. Then, for every $0 \leq i < m$, there exists $a_i
  \in A$ such that $((r_i, n_{i}), (r_{i+1}, n_{i+1})) \in R_{a_i}$.
  We can then show by induction on $i$, using (H2) and (H6), that
  $\commonk{A} \vp \in H(r_i, n_{i})$ holds for every $0 \leq < m$;
  hence, $\distrib{a_{i}} (\vp \con \commonk{A} \vp) \in H(r_i,
  n_{i})$.  Therefore, $\vp \in H(r_{i+1},n_{i+1})$ by (H2) and
  (H6). By taking $i = m-1$ we obtain $\vp \in H(r', n')$, as
  required.

  Now, assume $\neg \commonk{A} \vp \in H(r, n)$.  Then, the claim
  follows from (H7) and the inductive hypothesis, since $\rel{R}^C_A
  \subseteq \rel{R}'^C_A$.

  Suppose that $\chi = \vp \until \psi$. If $\vp \until \psi \in H(r,
  n)$, then the conclusion immediately follows from (H4) and the
  inductive hypothesis.  Suppose, on the other hand, that $\neg (\vp
  \until \psi) \in H(r, n)$.  Then, by (H2), $\neg \psi, \neg \vp \in
  H(r, n)$ or $\neg \psi, \next \neg (\vp \until \psi) \linebreak \in
  H (r, n)$.  In case former case, the inductive hypothesis
  immediately gives us the desired result.  In the latter, inductive
  hypothesis gives us $\notsat{M}{(r, n)}{\psi}$ and (H3) gives us
  $\neg (\vp \until \psi) \in H(r, n+1)$. Now the argument can be
  repeated.  Ultimately, using inductive hypothesis, we either get a
  finite path $(r, n), \ldots, (r, i)$ such that $\notsat{M}{(r,
    i)}{\vp}$ and $\notsat{M}{(r, j)}{\psi}$ holds for all $n \leq j
  \leq i$, or we get an infinite path $(r, n), (r, n+1), \ldots$ such
  that $\notsat{M}{(r, i)}{\psi}$ for all $i \geq n$.  In either case,
  $\notsat{M}{(r, n)}{\psi \until \psi}$.
\end{proof}

To show that satisfiability of a formula in a pseudo-TEM implies its
satisfiability in a TEM, we use a modification of the construction
from~\cite[Appendix A1]{FHV92} (see also~\cite{vdHM92}).

\begin{definition}
  \label{def:max_pahts}
  Let $\mmodel{M} = (\agents, S, R, \set{\rel{R}^D_A}_{A \in
    \powerne{\agents}}, \linebreak \set{\rel{R}^C_A}_{A \in
    \powerne{\agents}}, \ap, L)$ be a (pseudo-)TEM and let $r, r' \in
  R$ and $n, n' \in \nat$.  A \de{maximal path from $(r, n)$ to $(r',
    n')$} in \mmodel{M} is a sequence $(r,n) = (r_0, n_0), A_0, (r_1,
  n_1), \ldots, A_{m-1}, (r_m, n_{m}) \newline = (r', n')$ such that,
  for every $0 \leq i < m$, $((r_i, n_{i}), (r_{i+1}, n_{i+1}))
  \linebreak \in \rel{R}^D_{A_i}$, but $((r_i, n_{i}), (r_{i+1},
  n_{i+1})) \notin \rel{R}^D_{B}$ for any $B$ such that $A_i \subset B
  \subseteq \agents$.  A segment $\rho'$ of a maximal path $\rho$
  starting and ending with a point is a \de{sub-path} of $\rho$.
\end{definition}

\begin{definition}
  \label{def:reduced_paths}
  Let $\rho = (r_0, n_0), A_0 \ldots, A_{n-1}, (r_m, n_m)$ be a maximal
  path in \mmodel{M}.  The \de{reduction} of $\rho$ is obtained by,
  first, replacing in $\rho$ every longest sub-path $(r_p, n_p), A_p ,
  (r_{p+1}, \linebreak n_{p+1}) \ldots, A_{p+q-1}, (r_{p+q}, n_{p+q})$ such that  $r_p = r_{p+1} = \ldots = r_{p+q}$ with $r_p$ (i.e., eliminating
  loops) and, then, by replacing in the resultant path every longest
  sub-path $(r_j, n_j), A_j, (r_{j+1}, \linebreak n_{j+1}) \ldots, A_{j+m-1},
  (r_{j+m}, n_{j+m})$ such that $A_j = A_{j+1} = \ldots = A_{j+m-1}$ with
  $(r_j, n_{j}), A_j, (r_{j+m}, n_{j+m})$ (reducing multiple transitions
  along the same relation into a single transition).  A maximal path
  is \de{reduced} if it equals its reduction.
\end{definition}

\begin{definition}
  \label{def:forest-like-models}
  A (pseudo-)TEM \mmodel{M} is \de{forest-like} if, for every $r, r'
  \in R$ and every $n, n' \in \nat$, there exists at most one reduced
  maximal path from $(r, n)$ to $(r', n')$.
\end{definition}

One difference of the construction presented below from the one
in~\cite[Appendix A1]{FHV92} is that, instead of producing a tree-like
model, we rather produce a forest-like one, partly since every
``temporal level'' of the model we are going to build will not be
connected by epistemic relations to any other temporal level, and
partly because even within a single temporal level we will, in
general, construct more than one ``epistemic tree''.

\begin{lemma}
  \label{lem:psSTEM_to_STEM}
  If $\theta \in \lang$ is satisfiable in a (synchronous) pseudo-TEM,
  then it is satisfiable in a (synchronous) forest-like TEM.
\end{lemma}

\begin{proof} We will only consider the synchronous case, as it
  requires extra care.  Suppose that $\theta$ is satisfied in a
  synchronous pseudo-TEM $\mmodel{M} = (\agents, S, R,
  \set{\rel{R}^D_A}_{A \in \powerne{\agents}}, \linebreak
  \set{\rel{R}^C_A}_{A \in \powerne{\agents}}, \ap, L)$ at a point
  $(r, n)$.  To build a synchronous forest-like TEM $\mmodel{M}'$
  satisfying $\theta$, we use the modified tree-unraveling technique.
  First, every ``epistemic tree'' within a temporal level of the model
  will be made up of all \de{maximal} paths, rather than all paths, as
  in the standard tree-unraveling, since we want to ensure that paths
  between points are unique with respect to the relations
  $\rel{R}^D_A$ indexed by maximal coalitions, which will allow us to
  fix ``defects'' with respect to the $D$-relations. Second, every
  level will, in general, be made up of more than one epistemic tree,
  as every point at level $m \ne 0$ created as part of temporal run
  starting at a level $k < m$, will be a root of a separate tree.

  The construction starts by taking a submodel $\mmodel{M}_{(r, n)}$
  of $\mmodel{M}$ generated by the point $x = (r,n)$ at which $\theta$
  is satisfiable.

  Next, we define $\mmodel{M}'$ by recursion on the temporal levels.
  We view a level $k$ as partitioned into clusters \set{S^k_1, S^k_2,
    \ldots}, such that if $((r, k), (r',k)) \in S^k_i$, there is an
  (undirected) path along $D$-relations between $(r, k)$ and $(r',
  k)$.

  We start from level 0, corresponding to level $n$ in $\mmodel{M}$
  and level 0 in $\mmodel{M}_{(r, n)}$.  This level contains only one
  cluster $S^0$, generated by point $x$. In general, however, a level
  $k$ will have more than one cluster, so we describe the construction
  in more general terms.  At level $k$, for each cluster $S^k_i$, we
  choose arbitrarily a point $(r_i, k) \in S^k_i$ (at level $0$,
  however, we choose $x$); this point is going to be the root of an
  epistemic tree associated with that cluster.  Now, we call a maximal
  path $\rho$ in $\mmodel{M}$ a \emph{$(r_i, k)$-max-path} if the
  first component of $\rho$ is $(r_i, k)$.  We denote the last element
  of $\rho$ by $l(\rho)$. Notice that $(r_i, k)$ is by itself an
  $(r_i, k)$-max-path.  Now, let $\widehat{S}^k_i$ be the set of all
  $(r_i, k)$-max-paths in $\mmodel{M}$.  For every $A \in
  \powerne{\agents}$, let $\rel{R}^*{}^D_A = \crh{(\rho, \rho')}
  {\rho, \rho' \in \bigunion_{i}\, \widehat{S}^k_i$ and $\rho' = \rho,
    A, l(\rho')}$. Let, furthermore, $\rel{R}'^D_A$ to be the
  reflexive, symmetric, and transitive closure of
  $\rel{R}^*{}^D_A$. Notice that $(\rho,\rho')\in \rel{R}'^D_A$ holds
  iff one of the paths $\rho$ and $\rho'$ extends the other by a
  sequence of $A$-steps. Therefore, two different states in
  $\bigunion_{i}\,\, \widehat{S}^k_i$ can only connected by
  $\rel{R}'^D_A$ for at most one maximal coalition $A$.  Further, we
  stipulate the following \emph{downwards closure condition}: whenever
  $(\rho, \tau) \in \rel{R}'^D_A$ and $B \subseteq A$, then $(\rho,
  \tau) \in \rel{R}'^D_B$.  The relations $\rel{R}'^C_A$ are then
  defined as in any TEF.

  We next describe how to create level $m+1$ of $\mmodel{M}'$ assuming
  that level $m$ has already been defined.  First, carry out for $m+1$
  the construction described in the previous paragraph for an
  arbitrary level $k$.  Secondly, for every pair of states $\rho \in
  \bigunion_{i}\,\, \widehat{S}^{m}_i$ and $\tau \in \bigunion_{i}\,\,
  \widehat{S}^{m+1}_i$ make $(\tau, m+1)$ a temporal successor of
  $(\rho, m)$ if $l(\tau)$ is a such successor of $l(\rho)$ in
  \mmodel{M}.

  To complete the definition of $\mmodel{M}'$, we put $L'(\rho) =
  L(l(\rho))$ for every $\rho \in P(\mmodel{M}')$, where
  $P(\mmodel{M}')$ is the set of points of $\mmodel{M}'$.  It is clear
  from the construction, namely from the downward saturation condition
  above, that $\mmodel{M}'$ is a synchronous pseudo-TEM.  We now show
  that it is a TEM satisfying $\theta$.

  To prove the first part of the claim, we need extra terminology.  We
  call a maximal path $\rho_1, A_1, \rho_2, \ldots, A_{n-1}, \rho_n$
  in $\mmodel{M}'$ \de{primitive} if, for every $0 \leq i < n$, either
  $(\rho_i, \rho_{i+1}) \in \rel{R}^*{}^D_{A_i}$ or $(\rho_{i+1},
  \rho_i) \in \rel{R}^*{}^D_{A_i}$.  A primitive path $\rho_1, A_1,
  \rho_2, \ldots, \linebreak A_{n-1}, \rho_n$ is \de{non-redundant} if
  there is no $0 \leq i < n$ such that $\rho_i = \rho_{i+2}$ and $A_i
  = A_{i+1}$. Intuitively, in a non-redundant path we never go from a
  state $\rho$ (forward or backward) along a relation and then
  immediately back to $\rho$ along the same relation.  Since the
  relations $\rel{R}^*{}^D_A$ are edges of a tree, it immediately
  follows that ($S'$ denotes the state space of $\mmodel{M}'$):
  \begin{description}
  \item[(\ddag)] for every pair of states $\rho, \tau \in S'$, there
    exists at most one non-redundant primitive path from $\rho$ to
    $\tau$.
  \end{description}
  Lastly, we call a primitive path $\rho_1, A, \rho_2, \ldots, A,
  \rho_n$ an \de{$A$-primitive path}.

  We will now show that maximal reduced paths in $\mmodel{M}'$ stand
  in one-to-one correspondence with non-redundant primitive paths.  It
  will then follow from (\ddag) that maximal reduced paths between any
  two states of $\mmodel{M}'$ are unique, and thus $\mmodel{M}'$ is
  forest-like, as claimed.  Let $P = \rho_1, A_1, \ldots, \linebreak
  A_{n-1}, \rho_n$, where $\rho_1 = \rho$ and $\rho_n = \tau$, be a
  maximal reduced path from $\rho$ to $\tau$ in $\mmodel{M}'$. Since
  $(\rho_i, \rho_{i+1}) \in \rel{R}'^D_{A_i}$, there exists a
  non-redundant $A_i$-primitive path from $\rho_i$ to $\rho_{i+1}$,
  which in view of (\ddag) is unique. Let us obtain a path $P'$ from
  $\rho$ to $\tau$ by replacing in $\rho$ every link $(\rho_i, A_i,
  \rho_{i+1})$ by the corresponding non-redundant $A_i$-primitive path
  from $\rho_i$ to $\rho_{i+1}$.  Call $P'$ an \de{expansion} of $P$.
  In view of (\ddag), every path has a unique expansion.  Now, it is
  easy to see that $P$ is a reduction of $P'$.  Since the reduction of
  a given path is unique, too, it follows that there exists a
  one-to-one correspondence between reduced paths and non-redundant
  primitive paths in $\mmodel{M}'$.

  We now prove that $\rel{R}'^D_A = \biginter_{a \in A} \rel{R}'^D_a$
  for every $A \in \powerne{\agents}$, and hence $\mmodel{M}'$ is a
  TEM.  The left to right inclusion is immediate, as $\mmodel{M}'$ is
  pseudo-TEM.  For the other direction, assume that $((r, n), (r', n))
  \in \rel{R}'^D_a$ holds for every $a \in A$.  Then, for every $a \in
  A$, there exists a path, and therefore a maximal reduced path, from
  $(r, n)$ to $(r', n)$ along relations $\rel{R}'{}^D_{A'}$ such that
  $a \in A'$.  As $\mmodel{M}'$ is forest-like, there is only one
  maximal reduced path from $(r, n)$ to $(r', n)$.  Therefore, the
  relations $\rel{R}^D_{A'}$ linking $(r, n)$ to $(r', n)$ along this
  path are such that $A \subseteq A'$ for every $A'$.  Then, by the
  downwards closure condition, there is a path from $(r, n)$ to $(r',
  n)$ along the relation $\rel{R}'^D_A$ and, hence, $((r, n), (r', n))
  \in \rel{R}'^D_A$, as desired.

  Finally, it remains to prove that $\mmodel{M}'$ satisfies $\theta$.
  First, notice that $(\rho, \tau) \in \rel{R}'_A$ iff there exists an
  $A$-primitive path from $\rho$ to $\tau$. Hence, as every
  $\rel{R}_A$ is an equivalence relation, if $(\rho, \tau) \in
  \rel{R}'_A$, then $(l(\rho), l(\tau)) \in \rel{R}'_A$. It is now
  straightforward to check that the relation $Z = \crh{(\rho,
    l(\rho)}{\rho \in S'}$ is a bisimulation between $\mmodel{M}'$ and
  $\mmodel{M}$.  Since $(x, l(x)) \in Z$, it follows that
  \sat{M'}{x}{\theta}, and we are done.
\end{proof}

\begin{theorem}
  \label{thr:models_equal_hintikka}
  Let $\theta \in \lang$.  Then, $\theta$ is satisfiable in a TEM iff
  there exists a TEHS satisfying $\theta$.
\end{theorem}

\begin{proof}
  Immediate from Lemmas~\ref{lem:form_models_to_hintikka},
  \ref{lem:STEHS_to_psSTEM} and \ref{lem:psSTEM_to_STEM}.
\end{proof}

\section{Tableaux for \\ CMATEL(CD + LT)}
\label{sec:tableaux}

In the present section, we describe the tableau procedure for testing
formulae of \CMATELCDLT\ for satisfiability in synchronous systems, as
this case requires more care.  We then briefly mention how to modify
the procedure for asynchronous case and argue the the output of both
procedures for the same input formula is the same, implying the
equivalence of two semantics.

\subsection{Overview of the tableau procedure}

The tableau procedure for testing a formula $\theta \in \lang$ for
satisfiability attempts to construct a non-empty graph
$\tableau{T}^{\theta}$ (called \fm{tableau}), whose nodes are finite
subsets of \lang, representing \emph{sufficiently many} TEHSs, in the
sense that, if $\theta$ is satisfiable in a TEHS, it is satisfiable
in a one represented by a tableau for $\theta$. The philosophy
underlying our tableau algorithm is essentially the same as the one
underpinning the tableau procedure for \LTL\ from~\cite{Wolper85},
recently adapted to multiagent epistemic logics in~\cite{GorSh09}; 
this philosophy can be traced back to~\cite{Pratt80}. To make the
present paper self-contained, we outline the basic ideas behind our
tableau algorithm in line with those references. The particulars of
the tableaux presented here, however, are specific to \CMATELCDLT.

Usually, tableaux work by decomposing the input formula into simpler
formulae. In the classical propositional case, ``simpler'' implies
shorter, thus ensuring the termination of the procedure. The
decomposition into simpler formulae in the tableau for classical
propositional logic produces a tree representing an exhaustive search
for a Hintikka set (the classical analogue of Hintikka structures) for
the input formula $\theta$.  If at least one leaf of that tree
produces a Hintikka set for $\theta$, the search has succeeded and
$\theta$ is pronounced satisfiable; otherwise it is declared
unsatisfiable.

When applied to logics containing fixpoint-definable operators, such
as $\commonk{A}$ and $\until$, these two defining features of the
classical tableau method no longer apply. First, the decomposition of
fixpoint formulae, which is done by unfolding their fixpoint
definitions, produces larger formulae: $\commonk{A} \vp$ is decomposed
into formulae of the form $\distrib{a} (\vp \con \commonk{A} \vp)$,
while $\vp \until \psi$ is decomposed into $\psi$ and $\vp \con \next
(\vp \until \psi)$. Hence, we need a termination-ensuring
mechanism. In our tableaux, such a mechanism is provided by the use
(and reuse) of so called ``prestates'', whose role is to ensure the
finiteness of the construction and, hence, termination of the
procedure. Second, the only reason why a tableau may fail to produce a
Hintikka set for the input formula in the classical case is that every
attempt to build such a set results in a collection of formulae
containing a \emph{patent inconsistency}, i.e., a complementary pair
of formulae $\vp, \neg \vp$.  In the case of \CMATELCDLT, there are
other such reasons, as the tableaux in this case are meant to
represent TEHSs, which are more involved structures than classical
Hintikka sets.  One additional reason has to do with eventualities:
the presence of an eventuality $\neg \commonk{A} \vp$ in the label of
a state $s$ of a TEHS \hintikka{H} requires the existence in
\hintikka{H} of an $A$-path from $s$ to a state $t$ whose label
contains $\neg \vp$ (condition (H7) of Definition~\ref{def:hs}). An
analogous requirement applies to eventualities of the form $\vp \until
\psi$ due to condition (H4) of Definition~\ref{def:hs}.  The tableau
analogs of these conditions is called \emph{realization of
  eventualities}.  If a tableau contains nodes with unrealized
eventualities, then it cannot produce a TEHS, and thus it is
``bad''. The third possible reason for a tableau to be ``bad'' has to
do with successor nodes: it may so happen that some of the successors
of a node $s$ which are necessary for the satisfaction of $s$ are
unsatisfiable.  Notice that TEHSs, and consequently the associated
tableaux, contain two kinds of ``successor'' nodes: temporal and
epistemic.  The non-satisfiability of either kind of successor can
ruin the chances of a tableau node to correspond to a state of a TEHS.

The tableau procedure consists of three major phases: \fm{pretableau
  construction}, \fm{prestate elimination}, and \fm{state
  elimination}.  During the pretableau construction phase, we produce
a directed graph $\tableau{P}^{\theta}$---called the \emph{pretableau}
for $\theta$---whose set of nodes properly contains the set of nodes
of the tableau $\tableau{T}^{\theta}$ we are building.  The nodes of
$\tableau{P}^{\theta}$ are sets of formulae of two kinds: \fm{states}
and \fm{prestates}. States are fully expanded sets, meant to represent
(labels of) states of a Hintikka structure, while prestates play a
temporary role in the construction of $\tableau{T}^{\theta}$.  During
the prestate elimination phase, we create a smaller graph
$\tableau{T}_0^{\theta}$ out of $\tableau{P}^{\theta}$, called the
\fm{initial tableau for $\theta$}, by eliminating all the prestates
from $\tableau{P}^{\theta}$ and accordingly redirecting its edges.
Finally, during the state elimination phase, we remove from
$\tableau{T}_0^{\theta}$ all the states, if any, that cannot be
satisfied in a TEHS, either because they contain unrealized
eventualities or because they lack a necessary successor (patently
inconsistent states are removed ``on the fly'' during the state
creation stage). The elimination procedure results in a (possibly
empty) subgraph $\tableau{T}^{\theta}$ of $\tableau{T}_0^{\theta}$,
called the \de{final tableau for $\theta$}. If some state $\Delta$ of
$\tableau{T}^{\theta}$ contains $\theta$, we declare $\theta$
satisfiable; otherwise, we declare it unsatisfiable. The construction
of the tableau is illustrated in Example~\ref{TableauExample} given at
the end of Section~\ref{sec:state_elimination}.

\subsection{Pretableau construction phase}
\label{sec:construction}

All states and prestates of the pretableau $\tableau{P}^{\theta}$
constructed during this phase are ``time-stamped'', the notation
$\G^{[n]}$ indicating that prestate $\G$ was created as the $n$th
component of a run; analogously for states.

The pretableau contains three types of edge, described below.  As
already mentioned, a tableau attempts to produce a compact
representation of a sufficient number of TEHSs for the input formula,
which are the result of an exhaustive search for a TEHS satisfying
$\theta$. One type of edge, depicted by unmarked double arrows
$\brancharrow$, represents the search dimension of the tableau.
Exhaustive search considers all possible alternatives, which arise
when expanding prestates into states by branching when dealing with
the ``disjunctive formulae''. Thus, when we draw a double arrow from a
prestate \G\ to states $\D$ and $\D'$ (depicted as $\G \brancharrow
\D$ and $\G \brancharrow \D'$, respectively), this intuitively means
that, in any TEHS, a state whose label extends the set \G\ has to
contain at least one of $\D$ and $\D'$. Our first construction rule,
\Rule{SR}, prescribes how to create tableau states from prestates.

Given a set $\G \subseteq \lang$, we say that $\D$ is a \de{minimal
fully expanded extension of \G} if $\D$ is fully expanded, $\G
\subseteq \D$, and there is no $\D'$ such that $\G \subseteq \D'
\subset \D$ and $\D'$ is fully expanded.

\smallskip

\textbf{Rule} \Rule{SR} Given a prestate $\G^{[n]}$ such that \Rule{SR}
has not been applied to \Rule{SR} earlier, do the following:
\begin{enumerate}
  \itemsep-4pt
\item Add all minimal fully expanded extensions $\D^{[n]}$ of
  $\G^{[n]}$ that are not patently inconsistent as \de{states};
\item if $\D^{[n]}$ contains no formulae $\next \vp$, add
  $\next \truth$ to it;
\item for each so obtained state $\D^{[n]}$, put $\G^{[n]} \brancharrow
  \D^{[n]}$;
\item if, however, the pretableau already contains a state $\D'^{[m]}$
  that coincides with $\D^{[n]}$, do not create another copy of
  $\D'^{[m]}$, but only put $\G^{[n]} \brancharrow \D'^{[m]}$.
\end{enumerate}

We denote by $\st{\G^{[n]}}$ the  set of states
\crh{\D}{\G^{[n]} \brancharrow \D}. Note that we remove patently
inconsistent states ``on the fly'', thus never making them part of a
pretableau.

Notice that in all construction rules, as in \Rule{SR}, we allow reuse
of (pre)states, which were originally stamped with a possibly
different time-stamp. This does not correspond to one state or
prestate being part of two different runs, at different moments of
time (the absolute time is supposed to be the same in all runs, even
though agents may not be able to observe it, in asynchronous systems);
rather, the ``futures'' of these runs, starting from the reused
(pre)state can be assumed to be identical, modulo the time difference.

The second type of edge in a pretableau represents epistemic relations
in the TEHSs that the procedure attempts to build.  This type of edge
is represented by single arrows marked with epistemic formulae whose
presence in the source state requires the presence in the tableau of a
target state, reachable by a particular epistemic relation.  All such
formulae have the form $\neg \distrib{A} \vp$ (as can be seen from
Definition~\ref{def:hs}).  Intuitively if, say $\neg \distrib{A} \vp
\in \D^{[n]}$, then we need some prestate $\G^{[n]}$ containing $\neg
\vp$ to be accessible from $\D^{[n]}$ by $\rel{R}^D_A$ (notice that
the newly created prestates bear the same time stamp as the source
state; this reflects the fact that we are considering the synchronous
case).  The reason we mark these single arrows not just by a coalition
$A$, but by a formula $\neg \distrib{A} \vp$, is that we have to
remember not just what relation connects states whose labels contain
$\D^{[n]}$ and $\G^{[n]}$, but why we had to create this particular
$\G^{[n]}$.  This information will be needed when we start eliminating
prestates, and then states.  We now formulate the rule producing this
second type of edges in the pretableau.

\smallskip

\textbf{Rule} \Rule{DR}: Given a state $\D^{[n]}$ such that $\neg
\distrib{A} \vp \in \D^{[n]}$, $\D^{[n]}$ and \Rule{DR} has not been
applied to $\D^{[n]}$ earlier, do the following:
\begin{enumerate}
  \itemsep-4pt
\item Create a new prestate $\G^{[n]} = \set{\neg \vp} \union
  \bigunion_{A' \subseteq A} \crh{\distrib{A'} \psi}{\distrib{A'} \psi
    \in \D^{[n]}} \union \bigunion_{A' \subseteq A} \crh{\neg
    \distrib{A'} \psi}{\neg \distrib{A'} \psi \in \D^{[n]}}$;
\item connect $\D^{[n]}$ to $\G^{[n]}$ with $\stackrel{\neg \distrib{A}
    \vp}{\longrightarrow}$;
\item if, however, the tableau already contains a prestate $\G'^{[n]} =
  \G^{[n]}$, do not add another copy of $\G'^{[n]}$, but simply connect
  $\D^{[n]}$ to $\G'^{[n]}$ with $\stackrel{\neg \distrib{A}
    \vp}{\longrightarrow}$.
\end{enumerate}

Lastly, the third type of edge, depicted by single unmarked arrow
$\longrightarrow$, represents temporal transitions.
We now state the rule that creates such arrows.

\smallskip

\textbf{Rule} \Rule{Next}: Given a state $\D^{[n]}$ such that
\Rule{Next} has not been applied to $\D^{[n]}$ earlier, do the
following:
\begin{enumerate}
  \itemsep-4pt
\item Create a new prestate $\G^{[n+1]} = \crh{\vp}{\next \vp \in
    \D^{[n]}}$;
\item connect $\D^{[n]}$ to $\G^{[n+1]}$ with $\longrightarrow$;
\item if, however, the tableau already contains a prestate $\G'^{[m]}
  = \G^{[n+1]}$, do not add another copy of $\G'^{[m]}$, but
  simply connect $\D^{[n]}$ to $\G'^{[m]}$ with $\longrightarrow$.
\end{enumerate}

Note that, due to step 2 in \Rule{SR}, every state contains at least
one formula of the form $\next \vp$.

Having stated the rules, we now describe how the construction phase
works.  We start off by creating a single prestate \set{\theta},
where $\theta$ is the input formula.  Then we alternatingly apply
\Rule{DR} and \Rule{Next} to the prestates created at the previous
stage and then applying \Rule{SR} to the newly created states.
The construction state is over when the applications of \Rule{DR}
and \Rule{Next} do not produce any new prestates.

\subsection{Prestate elimination phase}

At this phase we remove from $\tableau{P}^{\theta}$ all the prestates
and double arrows, by applying the following rule:

\smallskip

\textbf{Rule} \Rule{PR} For every prestate $\G$ in
$\tableau{P}^{\theta}$, do the following:

\begin{enumerate}
  \itemsep-2pt
\item Remove $\G$ from $\tableau{P}^{\theta}$;
\item if there is a state $\D$ in $\tableau{P}^{\theta}$ with $\D
  \stackrel{\chi}{\longrightarrow} \G$, then for every state $\D' \in
  \st{\G}$, put $\D \stackrel{\chi}{\longrightarrow} \D'$;
\item if there is a state $\D$ in $\tableau{P}^{\theta}$ with $\D
  \longrightarrow \G$, then for every state $\D' \in \st{\G}$, put $\D
  \longrightarrow \D'$.
\end{enumerate}

The resulting graph, denoted $\tableau{T}_0^{\theta}$, is called the
\emph{initial tableau}.

\subsection{State elimination phase}
\label{sec:state_elimination}

During this phase we remove from $\tableau{T}_0^{\theta}$ states that
are not satisfiable in a TEHS.  There are two reasons why a state $\D$
of $\tableau{T}_0^{\theta}$ can turn out to be unsatisfiable: either
satisfiability of $\D$ requires satisfiability of some other
(epistemic or temporal) successor states which are unsatisfiable, or
$\D$ contains an eventuality that is not realized in the
tableau. Accordingly, we have three elimination rules (as two
different rules deal with epistemic and temporal successors):
\Rule{E1E}, \Rule{E1T}, and \Rule{E2}.

Formally, the state elimination phase is divided into stages; we start
at stage 0 with $\tableau{T}_0^{\theta}$; at stage $n+1$ we remove
from the tableau $\tableau{T}_n^{\theta}$ obtained at the previous
stage exactly one state, by applying one of the elimination rules,
thus obtaining the tableau $\tableau{T}_{n+1}^{\theta}$. We state the
rules below, where $S_m^{\theta}$ denotes the set of states of
$\tableau{T}_{m}^{\theta}$.

\smallskip

\Rule{E1E} If $\D \in S^{\theta}_n$ contains a formula $\chi = \neg
\distrib{A} \vp$ and $\D \stackrel{\chi}{\longrightarrow} \D'$ does
not hold for any $\D' \in S_n^{\theta}$, obtain
$\tableau{T}_{n+1}^{\theta}$ by eliminating $\D$ from
$\tableau{T}_n^{\theta}$.

\smallskip

\Rule{E1T} If If $\D \in S^{\theta}_n$ and $\D \longrightarrow \D'$
does not hold for any $\D' \in S^{\theta}_n$, obtain
$\tableau{T}_{n+1}^{\theta}$ by eliminating $\D$ from
$\tableau{T}_n^{\theta}$.

\smallskip

For the third elimination rule, we need the concept of
\emph{eventuality realization}.  We say that the eventuality $\neg
\commonk{A} \vp$ is realized at $\D$ in $\tableau{T}^{\theta}_n$ if
there exists a finite path $\D = \D_0, \D_1, \ldots, \D_m$ (where $m
\geq 0$) such that $\neg \vp \in \D_m$ and for every $0 \leq i < m$
there exist $\chi_i = \distrib{B} \psi_i$ such that $B \subseteq A$
and $\D_i \stackrel{\chi_i}{\longrightarrow} \D_{i+1}$. Analogously,
we say that the eventuality $\vp \until \psi$ is realized at $\D$ in
$\tableau{T}^{\theta}_n$ if there exists a finite path $\D = \D_0,
\D_1, \ldots, \D_m$ (where $m \geq 0$) such that $\psi \in \D_m$, and
for every $0 \leq i < m$, both $\D_i \longrightarrow \D_{i+1}$ and
$\vp \in \D_i$ hold.

\smallskip

\Rule{E2} If $\D \in S_n^{\theta}$ contains a (temporal or epistemic)
eventuality $\xi$ that is not realized at $\D$ in
$\tableau{T}_n^{\theta}$, then obtain $\tableau{T}_{n+1}^{\theta}$ by
removing $\D$ from $\tableau{T}_n^{\theta}$.

\smallskip

We check for realization of eventualities by running the following
iterative procedure that eventually marks all states that realize a
given eventuality $\xi$ in $\tableau{T}_n^{\theta}$.  If $\xi = \neg
\commonk{A} \vp$, then initially, we mark all $\D \in S_n^{\theta}$
such that $\neg \vp \in \D$. Then, we repeat the following procedure
until no more states get marked: for every still unmarked $\D \in
S_n^{\theta}$, mark $\D$ if there is at least one $\D'$ such that $\D
\stackrel{\distrib{B} \psi}{\longrightarrow} \D'$ for some $B\subseteq
A$ and $\D'$ is marked.  The procedure for eventualities of the form
$\vp \until \psi$ is analogous.

We have so far described individual rules and their implementation; to
describe the state elimination phase as a whole, we need to specify
the order of their application. We need to be careful, as having
applied \Rule{E2}, we could have removed all the states accessible
from some $\D$ either along the arrows marked with an epistemic
formula $\chi$ or along unmarked arrows $\longrightarrow$; hence, we
need to reapply \Rule{E1E} and \Rule{E1T} to the resultant tableau to
remove such $\D$'s. Conversely, having applied \Rule{E1E} and
\Rule{E1T}, we could have thrown away some states that were needed for
realizing certain eventualities; hence, we need to reapply
\Rule{E2}. Therefore, we need to apply \Rule{E2}, \Rule{E1E}, and
\Rule{E1T} in a dovetailed sequence that cycles through all the
eventualities.  More precisely, we arrange all eventualities occurring
in $\tableau{T}_0^{\theta}$ in a list $\xi_1, \ldots, \xi_m$.  Then,
we proceed in cycles. Each cycle consists of alternatingly applying
\Rule{E2} to the pending eventuality (starting with $\xi_1$), and then
applying \Rule{E1E} and \Rule{E1T} to the resulting tableau, until all
the eventualities have been dealt with, i.e., we reached $\xi_m$.
These cycles are repeated until no state is removed in a whole
cycle. Then, the state elimination phase is over.

The graph produced at the end of the state elimination phase is called
the \fm{final tableau for $\theta$}, denoted by $\tableau{T}^{\theta}$
and its set of states is denoted by $S^{\theta}$.

\begin{definition}
  The final tableau $\tableau{T}^{\theta}$ is \de{open} if $\theta \in
  \D$ for some $\D \in S^{\theta}$; otherwise, $\tableau{T}^{\theta}$
  is \de{closed}.
\end{definition}

The tableau procedure returns ``no'' if the final tableau is closed;
otherwise, it returns ``yes'' and, moreover, provides sufficient
information for producing a finite pseudo-model satisfying $\theta$;
that construction is sketched in Section
\ref{sec:soundness_completeness}.

\begin{example}
\label{TableauExample}
In this example, we show how our procedure works on the formula $\neg
\commonk{\set{a,b}} p \until\, \distrib {\set{a,c}} p$.  Below is the
complete pretableau for this formula.

\medskip

  
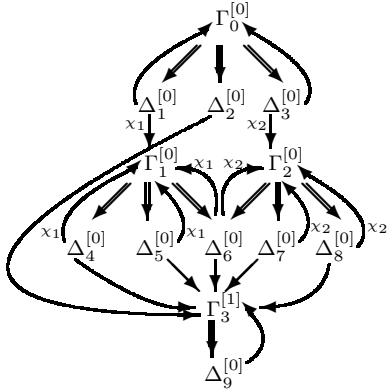
\begin{picture}(180,145)(-70,230)
   \footnotesize
    \thicklines

    \put(30,375){\makebox(0,0){
        {$\G^{[0]}_0$}
      }}
    
    \put(16,364){\line(-1,-1){10}}
    \put(17.5,364){\line(-1,-1){10}}
    \put(7.75,355){\vector(-1,-1){5}}

    \put(3,342){\makebox(0,0){
        {$\D^{[0]}_1$ }
      }}

    \put(-2,338){\vector(0,-1){13}}
    \put(-6,334){\makebox(0,0){
        {\tiny $\chi_1$ }
      }}

  \qbezier(-7, 342)(-12, 355)(20, 370)
  \put(17,368){\vector(1,1){5}}

  \put(24,364){\line(0,-1){10}}
  \put(25.25,364){\line(0,-1){10}}
  \put(24.75,355){\vector(0,-1){5}}
  
  \put(29,342){\makebox(0,0){
      {$\D^{[0]}_2$ }
    }}

  \qbezier(21, 337)(-130, 262)(17, 262)
  \put(12,262){\vector(1, 0){5}}

  \put(34,364){\line(1,-1){10}}
  \put(35.5,364){\line(1,-1){10}}
  \put(43.75,355){\vector(1,-1){5}}
  
    \put(50,342){\makebox(0,0){
        {$\D^{[0]}_3$ }
      }}

    \put(44,338){\vector(0,-1){13}}
    \put(40,334){\makebox(0,0){
        {\tiny $\chi_2$ }
      }}
    
    \qbezier(57, 342)(65, 353)(37, 369)
    \put(38,368){\vector(-1, 1){5}}

    
    \put(3,320){\makebox(0,0){
        {$\G^{[0]}_1$}
      }}

    \put(-9,312){\line(-1,-1){10}}
    \put(-10.5,312){\line(-1,-1){10}}
    \put(-18.75,303){\vector(-1,-1){5}}
    
    \put(-24,288){\makebox(0,0){
        {$\D^{[0]}_4$ }
      }}

    \qbezier(-34, 290)(-39, 303)(-7, 318)
    \put(-10,316){\vector(1,1){5}}

    \qbezier(-30, 283)(0, 262)(15, 265)
    \put(10,265){\vector(1, 0){5}}
    \put(-38, 293){\makebox(0,0){
        {\tiny $\chi_1$ }
      }}
  
    \put(-3.5,312){\line(0,-1){10}}
    \put(-2,312){\line(0,-1){10}}
    \put(-2.75,303){\vector(0,-1){5}} 
  
  \put(2,288){\makebox(0,0){
      {$\D^{[0]}_5$ }
    }}

  \qbezier(9, 290)(15, 298)(3, 310)
  \put(5,308){\vector(-1, 1){5}}
  \put(17, 293){\makebox(0,0){
      {\tiny $\chi_1$ }
    }}

  \put(7,312){\line(1,-1){10}}
  \put(8.5,312){\line(1,-1){10}}
  \put(16.75,303){\vector(1,-1){5}}
  
\cut{  \put(28, 288){\makebox(0,0){
        {$\D^{[0]}_6$ }
      }}
}


    \put(50,320){\makebox(0,0){
        {$\G^{[0]}_2$}
      }}

    \put(42,312){\line(-1,-1){10}}
    \put(43.5,312){\line(-1,-1){10}}
    \put(33.75,303){\vector(-1,-1){5}}

    \put(46,312){\line(0,-1){10}}
    \put(47.25,312){\line(0,-1){10}}
    \put(46.75,303){\vector(0,-1){5}}

    \put(57.5,312){\line(1,-1){10}}
    \put(56,312){\line(1,-1){10}}
    \put(65.75,303){\vector(1,-1){5}}

  \put(5, 283){\vector(1, -1){12}}

  \put(28, 288){\makebox(0,0){
        {$\D^{[0]}_6$ }
      }}

    \put(23, 283){\vector(0, -1){12}}

    \qbezier(23, 300)(26, 317)(9, 318)
    \put(13,317.5){\vector(-1, 0){5}}
    \put(20, 319){\makebox(0,0){
        {\tiny $\chi_1$ }
      }}

    \qbezier(26, 300)(23, 317)(40, 318)
    \put(36,317.5){\vector(1, 0){5}}
    \put(31, 319){\makebox(0,0){
        {\tiny $\chi_2$ }
      }}

  \put(48,288){\makebox(0,0){
      {$\D^{[0]}_7$ }
    }}

  \qbezier(55, 290)(63, 298)(51, 310)
  \put(53,308){\vector(-1, 1){5}}
  \put(64, 295){\makebox(0,0){
      {\tiny $\chi_2$ }
    }}

    \put(39, 283){\vector(-1, -1){12}}

  \put(70, 288){\makebox(0,0){
        {$\D^{[0]}_8$ }
      }}

    \qbezier(66, 281)(63, 265)(40, 265)
    \put(45,265){\vector(-1, 0){5}}

    \qbezier(77, 288)(86, 299)(58, 315)
    \put(59,314){\vector(-1, 1){5}}
    \put(85.5, 295){\makebox(0,0){
        {\tiny $\chi_2$ }
    }}


    \put(28, 265){\makebox(0,0){
        {$\G^{[1]}_3$ }
      }}

    \put(21,260){\line(0,-1){10}}
    \put(22.25,260){\line(0,-1){10}}
    \put(21.75,251){\vector(0,-1){5}}

    \put(28,240){\makebox(0,0){
        {$\D^{[0]}_9$ }
    }}

  \qbezier(35, 244)(47, 252)(35, 264)
  \put(38,262){\vector(-1, 1){5}}
  \end{picture}

  {\footnotesize $\chi_1 = \neg \distrib{a} (p \con
    \commonk{\set{a,b}} p)$; $\chi_2 = \neg \distrib{b} (p \con
    \commonk{\set{a,b}} p)$;

    $\G_0 = \set{\neg \commonk{\set{a,b}} p \until\, \distrib{\set{a,c}}
      p = \theta}$;

    $\D_1 = \set{\theta, \neg \commonk{\set{a,b}} p, \next \theta,
      \chi_1}$; $\D_2 = \set{\distrib{\set{a,c}} p, p, \next \truth
    }$;

    $\D_3 = \set{\theta, \neg \commonk{\set{a,b}} p, \next \theta,
      \chi_2}$; $\G_1 = \set{\chi_1, \neg (p \con \commonk{\set{a,c}}
      p)}$;

    $\G_2 = \set{\chi_2, \neg (p \con \commonk{\set{a,c}} p)}$; $\D_4
    = \set{\chi_1, \neg p, \next \top}$;

    $\D_5 = \set{\chi_1, \neg \commonk{\set{a, b}} p, \next \top}$;
    $\D_6 = \set{\chi_1, \neg \commonk{\set{a, b}} p, \chi_2, \next
      \top}$;

    $\D_7 = \set{\chi_2, \neg \commonk{\set{a,c}} p, \next \top}$;
    $\D_8 = \set{\chi_2, \neg p, \next \top}$;

    $\G_3 = \set{\truth}; \D_9 = \set{\truth, \next \truth}$.}

  \smallskip

  The initial tableau is obtained by removing all prestates (the
  $\G$s) and redirecting the arrows (i.e, $\D_1$ will be connected by
  unmarked single arrows to itself, $\D_2$, and $\D_3$).  It is easy
  to check that no states get removed during the state elimination
  stage; hence, the tableau is open and $\theta$ is satisfiable.
\end{example}

We now briefly mention how to modify the above procedure for the
asynchronous case.  The only difference occurs in the \Rule{DR} rule:
we now longer require that prestates produced during the application
of this rule to a given state $\D^{[n]}$ should have the same time
stamp as $\D$ (namely, $n$).  A brief analysis of the procedure shows
that this modification does not change the outcome of the procedure
for a given formula.  This, in particular, implies that the
satisfiability-wise equivalence of synchronous and asynchronous
semantics.

\section{Soundness, completeness,  \\ and  complexity}
\label{sec:soundness_completeness}

The \emph{soundness} of a tableau procedure amounts to claiming that
if the input formula $\theta$ is satisfiable, then the tableau for
$\theta$ is open.  To establish soundness of the overall procedure, we
use a series of lemmas showing that every rule by itself is sound; the
soundness of the overall procedure is then an easy consequence. The
proofs of the following three lemmas are straightforward.

\begin{lemma}
  \label{lm:expansion}
  Let $\G$ be a prestate of $\tableau{P}^{\theta}$ such that
  \sat{M}{(r,n)}{\G} for some TEM \mmodel{M} and point $(r,n)$.
  Then, \sat{M}{(r,n)}{\D} holds for at least one $\D \in \st{\G}$.
\end{lemma}

\begin{lemma}
  \label{lm:DR_sound}
  Let $\D \in S_m^{\theta}$, for $m \geq 0$, be such that
  \sat{M}{(r,n)}{\D} for some TEM \mmodel{M} and point $(r,n)$, and
  let $\neg \distrib{A} \vp \in \D$.  Then, there exists a point
  $(r',n') \in \mmodel{M}$ such that $((r,n), (r',n')) \in \rel{R}^D_A$
  and \sat{M}{(r',n')}{\D'} where $\D' = \set{\neg \vp} \union
  \bigunion_{A' \subseteq A} \crh{\distrib{A'} \psi}{\distrib{A'} \psi
    \in \D} \union \bigunion_{A' \subseteq A} \crh{\neg \distrib{A'}
    \psi}{\neg \distrib{A'} \psi \in \D}$.
\end{lemma}

\begin{lemma}
  \label{lm:Next_sound}
  Let $\D \in S_m^{\theta}$, for $m \geq 0$, be such that
  \sat{M}{(r,n)}{\D} for some TEM \mmodel{M} and a point $(r,n)$.
  Then, \sat{M}{(r,n+1)}{\next (\D)} where $\next (\D) =
  \crh{\vp}{\next \vp \in \D}$.
\end{lemma}

\begin{lemma}
  \label{lm:E3_sound1}
  Let $\D \in S_m^{\theta}$, for $m \geq 0$, be such that
  \sat{M}{(r,n)}{\D} for some TEM \mmodel{M} and a point $(r,n)$,
  and let $\neg \commonk{A} \vp \in \D$.  Then, $\neg \commonk{A} \vp$
  is realized at $\D$ in $\tableau{T}_m^{\theta}$.
\end{lemma}

\begin{proofidea}
  Since $\neg \commonk{A} \vp$ is true at $s$, there is a path in
  \mmodel{M} from $s$ leading to a state satisfying $\neg \vp$.  Since 
  the tableau performs exhaustive search, a chain of tableau
  states corresponding to those states in the model will be produced.
\end{proofidea}

The next lemma is proved likewise.

\begin{lemma}
  \label{lm:E3_sound1}
  Let $\D \in S_m^{\theta}$, for $m \geq 0$, be such that
  \sat{M}{(r,n)}{\D} for some TEM \mmodel{M} and a point $(r,n)$,
  and let $\vp \until \psi \in \D$.  Then, $\vp \until \psi$ is
  realized at $\D$ in $\tableau{T}_m^{\theta}$.
\end{lemma}

\begin{theorem}
  \label{thr:soundness}
  If $\theta \in \lang$ is satisfiable in a TEM, then
  $\tableau{T}^{\theta}$ is open.
\end{theorem}
\begin{proofsketch}
  Using the preceding lemmas, we show by induction on the number of
  stages in the state elimination phase that no satisfiable state can
  be eliminated due to any of the elimination rules.  The claim then
  follows from Lemma~\ref{lm:expansion}.
\end{proofsketch}

The \emph{completeness} of a tableau procedure means that if the
tableau for a formula $\theta$ is open, then $\theta$ is satisfiable
in a TEM. In view of Theorem~\ref{thr:models_equal_hintikka}, it
suffices to show that an open tableau for $\theta$ can be turned into
a TEHS for $\theta$.

\begin{lemma}
  \label{lm:open_tableau_hintikka}
  If $\tableau{T}^{\theta}$ is open, then a (synchronous) TEHS for
  $\theta$ exists.
\end{lemma}

\begin{proofsketch}
  The TEHS \hintikka{H} for $\theta$ is built by induction on the
  temporal levels, in order to take care of synchrony.  The main concern is to ensure that all  eventualities in the resultant structure are realized (all other properties of Hintikka structures easily transfer from an open tableau).  We alternate between realizing epistemic eventualities (formulae of the form $\neg \commonk{A} \vp$) and temporal eventualities (formulae of the form $\vp \until \psi$).

  We start by building the $0$th level of our prospective Hintikka
  structure from the level $0$ of the open tableau.  For each state
  $\D^{[0]}$ on this level, if $\D^{[0]}$ does not contain any
  epistemic eventualities, we define epistemic component for
  $\D^{[0]}$ to be $\D^{[0]}$ with exactly one successor reachable by
  $\neg \distrib{A} \psi$, for each $\neg \distrib{A} \psi \in
  \D^{[0]}$; if, on the other hand, $\neg \commonk{A} \vp \in
  \D^{[0]}$, then such a component is a tree obtained from a path in
  the tableau realizing $\neg \commonk{A} \vp$ at $\D^{[0]}$ by giving
  each component of the path ``enough'' successors, as described
  above.  We recursively repeat the procedure extending the current
  tree by attaching to its leaves associated components. As all the
  unrealized epistemic eventualities are propagated down the
  components (hence, appear in the leaves of the tree), we can stitch
  them up together to obtain a structure in which epistemic
  eventuality is realized.

  Now, having built the $0$th level of our prospective Hintikka
  structure, we take care of realizing all the temporal eventualities
  contained in the states of level $0$.  This is done exactly as in
  the completeness proof of the tableau procedure for \LTL: we define
  the temporal component for each $\D^{[0]}$ as follows: if $\D^{[0]}$
  does not contain any temporal eventualities, then we take $\D^{[0]}$
  with one of its temporal successors; otherwise, we take a temporal
  path realizing $\vp \until \psi \in \D^{[0]}$.  As eventualities are
  again passed down, we can stitch up an infinite, or ultimately
  periodic, path realizing all the eventualities contained in the
  states making up the path.

  Next, we repeat the procedure inductively.  For the $m$th epistemic
  level, we independently apply to each state on this level the
  procedure described above for level $0$, so that ``epistemic
  structures'' unfolding from any two points on level $m$ are
  disjoint, and also give to each newly created point a ``history''
  consisting of a path of $m-1$ states of the form \set{\truth} (so
  that we do not create any new epistemic eventualities at the levels
  we have already ``processed''). Having fixed all the epistemic
  eventualities at the $m$th level, we repeat the procedure described
  in the previous paragraph to fix all the temporal eventualities
  contained in states of level $m$.

  Thus, we produce a chain of structures ordered by
  inclusion. Eventually, we take the (infinite) union of all the
  structures defined at the finite states of that construction, and
  then put $H(\D^{[n]}) = \D^{[n]}$ for every $\D^{[n]}$, to obtain a
  TEHS for $\theta$.
\end{proofsketch}
\begin{theorem}[Completeness]
  \label{thr:completeness}
  Let $\theta \in \lang$ and let $\tableau{T}^{\theta}$ be open.
  Then, $\theta$ is satisfiable.
\end{theorem}
\begin{proof}
  Immediate from Lemma~\ref{lm:open_tableau_hintikka} and
  Theorem~\ref{thr:models_equal_hintikka}. 
\end{proof}

As for complexity, for lack of space, we only state that it runs
within exponential time (the calculation is routine). Therefore, the
\CMATELCDLT-satisfiability is in EXPTIME, which together with the
EXPTIME-hardness result from~\cite{HV89}, implies that it is
EXPTIME-complete.

\section{Concluding remarks}
\label{sec:concluding}

We developed an incremental-tableau based decision procedure for the
full coalitional multiagent temporal-epistemic logic of linear time
\CMATELCDLT. In this case, there is no essential interaction between
the temporal and the epistemic dimensions, which makes the tableau
construction easier to build and less expensive to run, by reducing it
to a combination of tableaux for \LTL\ and for the (epistemic) logic
\CMAELCD\ developed in \cite{GorSh09}.  We are convinced that our
procedure is---besides being rather intuitive---practi\-cally much
more efficient than the top-down tableaux, e.g., developed for a
fragment of our logic in~\cite{HM92}, and hence better suited to both
manual and automated execution. It is also easily amenable to
modifications suited to reasoning about subclasses of distributed
systems, e.g., those with a unique initial state. The branching time
case, which will be considered in a sequel to this paper, is
essentially a combination of tableaux for \CTL\ with those for
\CMAELCD.  On the other hand, the development of tableau-based
procedures for those logics from~\cite{HV89} whose satisfiability
problem has \mbox{EXPSPACE} lower bound is an open challenge.

\bibliographystyle{abbrv}

\end{document}